\documentclass{amsart}
\usepackage{wrapfig}
\usepackage{amsmath}
  \usepackage{paralist}
  \usepackage{graphics} 
  \usepackage{epsfig} 
\usepackage{graphicx}  \usepackage{epstopdf}
 \usepackage[colorlinks=true]{hyperref}
   \hypersetup{urlcolor=blue, citecolor=red}

\newtheorem{theorem}{Theorem}[section]

\newtheorem{lemma}[theorem]{Lemma}
\newtheorem{proposition}[theorem]{Proposition}

\newtheorem{definition}[theorem]{Definition}
\newtheorem{remark}[theorem]{Remark}
\newtheorem{observation}[theorem]{Observation}

\newcommand\bF{\mathbb F} \newcommand\bP{\mathbb P}
\newcommand\bK{\mathbb K} \newcommand{\bbsm}{\left (\begin{smallmatrix}}      \newcommand{\besm}{\end{smallmatrix}\right )}
\newcommand\beq{\begin{equation}}  \newcommand\eeq{\end{equation}}
\newcommand\beqs{\begin{equation*}}  \newcommand\eeqs{\end{equation*}}
\newcommand\mC{\mathcal C}  \newcommand\mG{\mathcal G} \newcommand\mP{\mathcal P}

\title{Counting generalized Reed--Solomon codes}

\author{Peter Beelen, David Glynn, Tom H\o holdt, and Krishna Kaipa}


\begin{document}

\maketitle

\centerline{\scshape Peter Beelen}
\medskip
{\footnotesize
 \centerline{Department of Applied Mathematics and Computer Science,}
   \centerline{Technical University of Denmark}
} 

\medskip

\centerline{\scshape David Glynn}
\medskip
{\footnotesize
 \centerline{School of Computer Science, Engineering and Mathematics,}
   \centerline{Flinders University Australia}
   \centerline{King Abdulaziz University, Jeddah, Saudi Arabia}
} 

\medskip

\centerline{\scshape Tom H\o holdt}
\medskip
{\footnotesize
 \centerline{Department of Applied Mathematics and Computer Science,}
   \centerline{Technical University of Denmark}
}

\medskip

\centerline{\scshape Krishna Kaipa}
\medskip
{\footnotesize
 \centerline{Department of Mathematics,}
   \centerline{IISER Pune}
}

\bigskip

\begin{abstract}
In this article we count the number of $[n,k]$ generalized Reed--Solomon (GRS) codes, including the codes coming from a non-degenerate conic plus nucleus. We compare our results with known formulae for the number of $[n,3]$ MDS codes with $n=6,7,8,9$.
\end{abstract}

\section{Introduction}

At the AGCT-India-2013 conference the fourth author presented a paper: ``An asymptotic formula in $q$ for the number of $[n,k]$ $q$-ary MDS codes'', extending previous work by S.R.~Ghorpade and G.~Lachaud \cite{GL}. A natural question is: ``How many of the MDS codes are generalized Reed--Solomon (GRS) codes?''. An MDS code is a certain type of linear $[n,k,d]_q$-code of length $n$, dimension $k$, and mimimum Hamming distance $d$, with symbols coming from the Galois field $\bF_q$. ``MDS'' stands for ``maximum distance separable''; that is, the Singleton bound $d\le n-k+1$ is met so that for MDS codes, $d = n-k+1$. Every linear code over $\bF_q$ with parameters $[n,k,d]_q$ has a $k\times n$ generator matrix $A$  over $\bF_q$, so that the codewords are the vectors in the row-space of $A$. Being MDS is equivalent to $A$ having every maximal $k\times k$ submatrix nonsingular; see \cite{Glynn1,singleton}. There are a number of ways of looking at and defining generalized Reed--Solomon (GRS) codes: e.g. they can be defined by the evaluations of certain polynomials. Here we take the point of view that a GRS code over $\bF_q$ of length $n\le q+1$ and dimension $k$ is defined by some generator matrix, the columns of which are a set of $n$ distinct points on a normal rational curve of degree $k-1$ in $\bP^{k-1}(\bF_q)$, the projective space of dimension $k-1$ over $\bF_q$. If $q$ is a power of two, we define a GRS code of dimension $3$ and length $q+2$ by a generator matrix, the columns of which are a set of $q+2$ distinct points on a non-degenerate conic plus its nucleus. For a definition of the nucleus of a conic, see \cite[Def.~1.28]{HKT} or \cite{Glynn2} for more general ways to construct nuclei algebraically. Thus the dual of the GRS $[q+2,3]$ code is MDS with parameters $[n,k]=[q+2,q-1]$, and also considered to be of GRS type.

Since the GRS type is preserved under duality, codes and their duals are paired. Hence the number of codes of GRS type for $[n,k]$ is the same for the dual parameters $[n,n-k]$. For dimensions $k=1$ and $k=n-1$ it is easy to see that any MDS code is a GRS code and that the number of such codes equals $(q-1)^{n-1}$: for $k=1$ the code is generated by a general vector of all non-zeros with first entry $1$. For $k=n-1$ the code is the dual of that code. For other values of $k$ the number of $[n,k]$ GRS codes appears not to be known, so in this paper we address this question and prove that for $4 \le k+2 \le n \le q+1$ the number of generalized Reed--Solomon codes is $$(q-1)^{n-1}(q-2)\cdots (q-n+2)$$ and that for $q=2^e$, with $e \ge 3$ the number of $[q+2,3]$ GRS codes is $$(q+2)(q-1)^{q+1}(q-2)!$$
We then compare the number of GRS codes with the number of MDS codes for small length and dimension three and observe some phenomena regarding the ratio between these numbers. Lastly, we give an explanation of these observations.

\section{Counting GRS codes of length up to $q+1$}

The main theorem we wish to prove in this section is the following:

\begin{theorem}\label{thm:countGRS}
Let $\mathbb{F}_q$ be the finite field with $q$ elements and choose natural numbers $k,n$ such that $4 \le k+2 \le n \le q+1$. Then the number $\gamma_{GRS}(k,n)$ of distinct GRS codes over $\mathbb{F}_q$ of length $n$ and dimension $k$ is equal to $$\gamma_{GRS}(k,n)=(q-1)^{n-1}\cdot (q-2) \cdots (q-n+2).$$
\end{theorem}

Our proof will use elements from group theory and finite geometry. While primarily working over $\mathbb{F}_q$, the finite field with $q$ elements, several of the statements on group theory are valid more generally. Therefore throughout this section $\bK$ denotes an arbitrary field. The notation $\bK^{\times}$ denotes the multiplicative group of non-zero elements of $\bK$. Let $c(t)$ denote the column vector $(1, t, t^2, \dots, t^{k-1})^T$ if $t \neq \infty$ and let $c(\infty) = (0, 0, \cdots, 0, 1)^T$. We consider $k \times n$ matrices that are the product of two matrices as follows:
\beq \label{eq:gen_mat}
   G_k(t, d) = \left[ c(t_1) \,c(t_2) \,\dots \,c(t_n) \right]\, \text{diag}(d_1, d_2, \cdots, d_n)
\eeq
where $t_1, \dots, t_n$ are $n$ distinct points on $\bP^1 = \bK\cup \{\infty\}$ and   $d_i \in \bK^{\times}$. Here $t$ and $d$ denote the tuples $(t_1,\dots,t_n)$ and $(d_1,\dots,d_n)$.
Let $S_{k,n}$ denote the set of  matrices of the form \eqref{eq:gen_mat}. We define an equivalence relation on $S_{k,n}$ as follows:
\[ G_k(t,d) \sim G_k(\tau,\delta) \text{ if there exists }  A \in GL_k(\bK) \text{ such that }  A \, G_k(t,d) = G_k(\tau,\delta)\]
In other words two  matrices of the form \eqref{eq:gen_mat} are equivalent if and only if they are row equivalent.

The following result is crucial:
\begin{theorem}  \label{main_thm}
    If $n \geq k+2$, then each $\sim$-equivalence class in $S_{k,n}$ is in bijective correspondence with a group $\mG$ which is a particular central extension of $PGL_2(\bK)$ by $\bK^{\times}$.
\end{theorem}
The group $\mG$ equals $\pi^{-1}(\rho'(PGL_2(\bK))$ where $\pi:GL_2(\bK) \to PGL_2(\bK)$ is the natural quotient homomorphism, and where $\rho': PGL_2(\bK) \to PGL_k(\bK)$ is a homomorphism defined in \eqref{eq:rho'_def}.
When $\bK$ is the field $\mathbb{F}_q$, elements of $S_{k,n}$ correspond to generator matrices of $q$-ary GRS codes of length $n$ and dimension $k$.
It is also clear that:
 \beq \label{eq:num_gen_mat}
  |S_{k,n}(\mathbb{F}_q)| = (q+1)q(q-1)\;  [(q-2)(q-3) \cdots(q-n+2)] \; (q-1)^n.
\eeq
Moreover two elements of $S_{k,n}$ are in the same $\sim$ equivalence class if and only if they generate the same code. When $\bK = \mathbb{F}_q$, any central extension of $PGL_2(\mathbb{F}_q)$ by $\mathbb{F}_q^{\times}$ has cardinality $|PGL_2(\mathbb{F}_q)| \cdot |\mathbb{F}_q^{\times}|= (q+1)q(q-1)^2$.
Therefore Theorem \ref{thm:countGRS} is a direct consequence of Theorem \ref{main_thm}, which is why we now focus on proving the latter theorem.

We start with  a lemma about the matrices $G_k(t,d)$. The proof is familiar in coding theory (at least when no $t_i = \infty$)  from the result that the dual of a GRS code is itself GRS.
\begin{lemma} \label{grs_dual}
  For $\delta = (\delta_1, \dots, \delta_n)$ as defined below,  the rows of $G_{n-k}(t,\delta)$ form a basis for the null space of $G_k(t,d)$.
    \beq    \label{eq:delta}
                 \delta_i  = \begin{cases} d_i^{-1} \cdot \displaystyle\prod_{ \{ j: j \neq i, t_j \neq \infty\} } (t_i  - t_j)^{-1} &\mbox{if }\,  t_i \neq \infty \\
                 - d_i^{-1} \cdot\displaystyle\sum_{j \neq i}          t_j^{n-2} d_j \delta_j  &\mbox{if }\, t_i = \infty. \end{cases}
    \eeq
\end{lemma}
\begin{proof}
First assume no $t_i  = \infty$. In this case the assertion  about $G_{n-k}(t,\delta)$ is equivalent to the assertion that
$\sum_i f(t_i) d_i \delta_i t_i^j  = 0$ for $0 \leq j \leq n-k-1$ for all polynomials $f \in \bK[t]$ of degree at most $k-1$. This in turn is equivalent to the fact that the coefficients of $t^{\ell}$ for $k \leq \ell \leq n-1$ in the following Lagrange interpolation form of the polynomial $f(t)$ vanish:
\[ f(t) = \sum_{i=1}^n f(t_i) d_i \delta_i \textstyle\prod_{j \neq i} (t - t_j). \]
In general, suppose $t_{\ell} = \infty$.  Let $G'$ be the $(k-1) \times (n-1)$ matrix obtained by deleting the last row and the $\ell$-th column of  $G_k(t,d)$. Let $G''$ be the $(n-k) \times (n-1)$ matrix obtained by deleting the $\ell$-th column of $G_{n-k}(t,d)$. It follows from the above discussion  that the rows of $G''$ are in the null space of $G'$. Already this implies that except the last row, the remaining rows of $G_{n-k}(t,\delta)$ are in the null space of $G_k(t,d)$. As regards the last row of $G_{n-k}(t,\delta)$, the choice of $\delta_{\ell}$ made in the hypothesis ensures that this row is also in the null space of  $G_{k}(t,d)$. Since the row space of $G_{n-k}(t,\delta)$ and the null space of $G_k(t,d)$ are $n-k$ dimensional, we have shown that the rows of $G_{n-k}(t,\delta)$ form a basis for the null space of $G_k(t,d)$.
\end{proof}

Next we recall the definition of a normal rational curve and we collect some  properties of this curve that we will need.  We use the notation $[x_0,x_1,\dots,x_{k-1}]$ for homogeneous coordinates of a point in $\bP^{k-1}$.
\begin{definition}
  The normal rational curve $\mC_k$ is the image $\epsilon_k(\bP^1)$ of the map  $\epsilon_k: \bP^1 \to \bP^{k-1}$ given by
  \beq \label{eq:rnc_def}
   \epsilon_k[1,t] =  [1,t,t^2, \dots,t^{k-1}] \quad \text{ i.e. } \;\epsilon_k [x,y] = [x^{k-1}, x^{k-2}y, \dots,  x y^{k-2}, y^{k-1}]
\eeq
 In general, for any $A \in PGL_k(\bK)$, the image of $A(\mC_k)$  is also called a normal rational curve.
   In the  sequel, we abbreviate the term normal rational curve as NRC.
   Throughout this work, we make the assumption that $k < |\bK|$. Under this assumption, any NRC has at least $k+2$ points.
\end{definition}

We  also need the intrinsic definition of a NRC. Let $V = \bK^m$.
Given $f_1, \dots, f_{k-1} \in V^*$ and an element $v \in V$, the map $(f_1, \dots, f_{k-1}) \mapsto f_1(v)f_2(v) \dots f_{k-1}(v)$ is symmetric and multilinear in $f_1, \dots, f_{k-1}$ and hence defines a map $V \to  \text{Sym}^{k-1}(V^*)^*$. The projectivization of this map is the $(k-1)$-th Veronese embedding
\beq \label{eq:NRC_intr} \eta_k: \bP(V) \to  \bP\left( (\text{Sym}^{k-1}(V^*))^* \right). \eeq
As pointed out in \cite[p.25, p.101]{Harris} this is the correct intrinsic description of the Veronese embedding in positive characteristic. In positive characteristic
there is no natural isomorphism between $(\text{Sym}^{k-1}(V^*))^*$ and $\text{Sym}^{k-1}(V)$. There is a graded algebra analogous to the graded algebra Sym$(V)$ (and isomorphic to it in characteristic zero) known as the divided power algebra of $V$, whose $k$-th graded component is $(\text{Sym}^{k-1}(V^*))^*$.  (See for example \cite[pp.566,587--588]{Eisenbud}).
When dim$(V)=2$, the image of the map \eqref{eq:NRC_intr} is the NRC. To see this, let $\{e, f\}$ denote the standard basis of $V=\bK^2$ and let $\{E, F\}$ denote the associated dual basis of $V^*$. The corresponding basis for Sym$^{k-1}(V^*)$ is $\{E^{k-i}F^{i-1} \, : \, i = 1 \cdots k\}$. Corresponding to this basis of Sym$^{k-1}(V^*)$  let $\{b_1, \dots, b_k\}$ denote the dual basis of $(\text{Sym}^{k-1}(V^*))^*$ (defined by $b_j(E^{k-i}F^{i-1}) = \delta_{ij}$ where $\delta_{ij}$ is the Kronecker delta function).
In terms of homogeneous coordinates with respect to the bases
$\{e,f\}$ and $\{b_1, \dots, b_k\}$ the map $\eta_k$ is described by
\beq \label{eq:eta_def}\eta_k(x e + y f) = (x^{k-1}, x^{k-2}y, \dots,  x y^{k-2}, y^{k-1}) \eeq
which is the same as \eqref{eq:rnc_def}.\bigskip

Let $V = \bK^2$. Let $\{e,f\}$ be the standard basis for $V$, and let  $\{e^{k-i} f^{i-1} \,:\, i=1\dots k\}$
, $\{E^{k-i} F^{i-1} \,:\, i=1\dots k\}$, and $\{b_1, \dots, b_k\}$ be the associated  bases for Sym$^{k-1}(V)$, Sym$^{k-1}(V^*)$ and $(\text{Sym}^{k-1}(V^*))^*$ as defined above. The group  $GL(V)$ acts on Sym$^{k-1} (V)$ by the $(k-1)$-th symmetric power representation $g \cdot v_1 v_2 \dots v_{k-1} =   gv_1  gv_2 \dots gv_{k-1}$. Let $\rho(g) \in GL_k(\bK)$ be the matrix representing this action with respect to the basis  $\{e^{k-i} f^{i-1} \,:\, i=1\dots k\}$.  Similarly $GL(V)$ acts on Sym$^{k-1}(V^*)$ by   $g \cdot  f_1 f_2 \dots f_{k-1} =   (f_1 \circ g)  (f_2 \circ g)   \dots (f_{k-1} \circ g)$.  The matrix of this action with respect to the basis  $\{E^{k-i} F^{i-1} \,:\, i=1\dots k\}$ is $\rho(g^{-t})$ where $g^{-t}$ is the inverse transpose of $g \in GL_2(\bK)$. It follows that the induced action of $GL(V)$ on $(\text{Sym}^{k-1}(V^*))^*$ has matrix $(\rho(g^{-t}))^{-t} = \rho(g^t)^t \in GL_k(K)$. We define:
\beq  \label{eq:rho'_def}
       \rho': PGL_2(\bK) \to PGL_k(\bK), \; \text{ given by } \; \rho'([g]) = [ \rho(g^t)^t)]
\eeq
where $[h]$ denotes the class in $PGL_m(\bK)$ of $h \in GL_m(\bK)$.
It follows from \eqref{eq:NRC_intr} and \eqref{eq:eta_def} that
$\rho'([g]) \cdot \epsilon_k [v]= \epsilon_k ( g \cdot [v])$. In other words:

\begin{proposition} \label{equivariance}
 For $g = \bbsm \alpha & \beta \\ \gamma & \delta \besm$,
\beq \label{eq:mobius}
   \rho'([g]) \cdot [1,t,\dots,t^{k-1}] = [1,\tau,\dots,\tau^{k-1}], \; \text{ where } \; \tau = \frac{\gamma + \delta t}{\alpha+\beta t}
  \eeq
 \end{proposition}

As an example, we see that for $k=2$, and $g  = \bbsm \alpha & \beta \\ \gamma & \delta \besm$,
  \[ \rho'([g]) = \left[ \bbsm \alpha^2 & 2 \alpha \beta & \beta^2\\  \alpha \gamma & \alpha \delta + \gamma \beta &   \beta \delta \\ \gamma^2 & 2 \gamma \delta & \delta^2 \besm \right],
\]
and $\rho'([g])[1,t,t^2] = [1,\tau,\tau^2]$ for $\tau = (\gamma + \delta t)/(\alpha+\beta t)$.

\begin{proposition} \label{mono}
		$ \rho': PGL_2(\bK) \to PGL_k(\bK)$ is a monomorphism.
\end{proposition}
\begin{proof}
Let $g  = \bbsm \alpha & \beta \\ \gamma & \delta \besm \in GL_2(\bK)$ with $[g] \in \text{ker}(\rho')$. Using Proposition \ref{equivariance}, we see that $t \mapsto  \tfrac{\gamma + \delta t}{\alpha+\beta t}$, i.e.  $[g]$ is the identity element of $PGL_2(\bK)$.
\end{proof}
The next result is well known if $\bK$ is algebraically closed (for example  \cite[\S4.3 p.530]{GH}) as well as if $\bK$ is  finite  (see \cite[Theorem 21.1.1 (v)]{Hirschfeld1985}).
\begin{theorem} \label{Castelnuovo}
  Let $\bK$ be an arbitrary field. A collection of  $k+2$ points in general position in $\bP^{k-1}$ lie on a unique NRC.
\end{theorem}
Here, general position means that no $k$ points of $\mP$ lie on any hyperplane of $\bP^{k-1}$.
We provide a proof for arbitrary fields $\bK$ using Lemma \ref{grs_dual}.

\begin{proof}
Up to a projective transformation, we can assume that the given collection $\mP$ of $k+2$ points is represented by the columns of the
$k \times (k+2)$ matrix $G=[I_k | v \, w]$, where all $1 \times 1$  and all $2 \times 2$ minors of the $k \times 2$ matrix $[v \, w]$ are nonzero. Let $t_i = v_i/w_i$ and consider the $2 \times (k+2)$ matrix  $G_2(t,d)$ with $d=(-v_1, -v_2,\dots, -v_k,1,1)$, and $t = (t_1,t_2,\dots,t_k,0,\infty)$, in other words:
\[ G_2(t,d) = \bbsm -v^t & 1 &0\\ -w^t &0 &1 \besm \]
Clearly the rows of $G_2(t,d)$ form  a basis for the null space of $G$. However, by Lemma \ref{grs_dual}, the rows of $G_2(t,d)$ also form a basis for the null space of
$G_k(t,\delta)$, where $\delta$ is related to $d$ by \eqref{eq:delta}. Since, the standard dot product on $\bF_q^{k+2}$ is non-degenerate, the subspace of $\bF_q^{k+2}$ orthogonal to the row space of $G_2(t,d)$ is the row space of $G$ as well as the row space of $G_k(t,\delta)$. It follows that there exists a non-singular matrix $R$ such that $G = R G_k(t,\delta)$.
In other words the points of $\mP$ lie on the NRC obtained by applying the projective transformation $[R]$ to the standard NRC in $\bP^{k-1}$. \\

Next, we establish uniqueness: Suppose $\mP$ lies on two NRCs. We may suppose that one of these is the standard NRC $\epsilon_k(\bP^1)$. In other words there exist:  $R \in GL_k(\bK)$,  distinct elements $T_1, \dots, T_{k+2} \in \bK \cup \{\infty\}$,  distinct elements $t_1, \dots, t_{k+2} \in \bK \cup \{\infty\}$,  and $D_1, \dots, D_{k+2} \in \bK^{\times}$ such that:
\beq \label{eq:two_rnc}
     R G_k(t,d) = G_k(T,D), \quad \text{where } \, d_1 = \dots = d_{k+2} = 1.
\eeq
We must show $T=t$. Using \eqref{eq:mobius}, and the fact that $PGL_2(\bK)$ acts triply transitively on $\bK \cup \{\infty\}$, we may assume $t_i = T_i$ for $i = 1 \dots 3$. By Lemma \ref{grs_dual}, there exists $d_1', \dots,d_{k+2}' \in \bK^{\times}$ and $D_1', \dots, D_{k+2}' \in \bK^{\times}$ such that the rows of $G_2(t,d')$ as well as the rows of $G_2(T,D')$ form a basis for the null space of  $G_k(T,D) = R G_k(t,d)$. Therefore, there exists $R' \in GL_2(\bK)$ such that
$R' G_2(t,d') = G_2(T,D')$. Since  $t_i = T_i$ for $i = 1 \dots 3$, the first three columns of $G_2(t,d')$ are pairwise independent eigenvectors of
$R'$. So, there must be an eigenvalue of $R'$ with two independent eigenvectors, i.e. $R'$ is a scalar matrix, whence $T=t$.
\end{proof}
\begin{definition}
    Let Aut$(\mC_k) = \{ A \in PGL_k(\bK) \mid  A (\mC_k)  = \mC_k\}$.
\end{definition}
\begin{lemma} \label{k+2}
   If  $A \in PGL_k(\bK)$  carries $k+2$ points of $\mC_k$ to some other set of $k+2$ points of $\mC_k$, then $A \in \text{Aut}(\mC_k)$.
\end{lemma}
\begin{proof}
Let $\mP$ be a set of $k+2$ points of $\mC_k$, such that $A \mP \subset \mC_k$. By Theorem \ref{Castelnuovo}, the unique NRC passing through $A \mP$ is $\mC_k$. Hence $A(\mC_k) = \mC_k$ i.e. $A \in $ Aut$(\mC_k)$.
\end{proof}
The next theorem states that the (linear) automorphism group of a NRC is isomorphic to $PGL_2(\bK)$.\\
\emph{Remark:} The NRC is the simplest case (dim$(V) = 2$) of the Veronese embedding  \eqref{eq:NRC_intr}. We expect a similar result for the Veronese variety over an arbitrary field : Aut$(\eta_k(\bP^{m-1})) = \rho'(PGL_{m}(\bK))$. Usually, automorphism groups of such varieties are computed in literature only for algebraically closed $\bK$. However, for the Pl\"ucker embedding of the Grassmannian, orthogonal and symplectic grassmannians, the automorphism groups were determined by Chow for arbitrary fields $\bK$, from a combinatorial/incidence geometry viewpoint in the well known work \cite{Chow}. We are not aware of any such results for the Veronese embedding. Therefore, we give a proof here which holds for arbitrary fields. For $\bK=\mathbb{F}_q$, this theorem can be found in e.g. \cite[Thm. 6.32]{HT2}.
\begin{theorem} \label{aut_rnc} (Automorphism group of the normal rational curve.) \\
   Let $\bK$ be an arbitrary field, and let $\mC_k$ be the NRC:  $\epsilon(\bP^1) \subset  \bP^{k-1}$ as defined in \eqref{eq:rnc_def}.
   Assuming  $k< |\bK|:$
   \[\text{Aut}(\mC_k) = \rho'(PGL_2(\bK)) \simeq PGL_2(\bK).\]
\end{theorem}

\begin{proof}
Let $R \in GL_k(\bK)$ such that $\pi(R) \in \text{Aut}(\mC_k)$. Here $\pi:GL_k(\bK) \to PGL_k(\bK)$ is the quotient homomorphism.   Let $t_1, t_2, \dots, t_{k+2}$ be distinct elements of $\bK \cup \{\infty\}$. Let $d_i = 1$ for $i = 1 \dots k+2$.
Since $\pi(R) \in \text{Aut}(\mC_k)$, there exist  $T_1, \dots, T_{k+2} \in \bK \cup \{\infty\}$, and $D_1, \dots, D_{k+2} \in \bK^{\times}$ such that $R G(t,d) = G(T,D)$ (this is the same as \eqref{eq:two_rnc}). By Proposition \ref{equivariance}, certainly $\rho'(PGL_2(\bK)) \subset \text{Aut}(\mC_k)$. Therefore, as in the proof of Theorem \ref{Castelnuovo} above, we may assume $t_i = T_i$ for $i= 1 \dots 3$. The discussion following \eqref{eq:two_rnc} in the proof of  Theorem \ref{Castelnuovo}, shows that $T = t$. We must now show that $R$ is a scalar matrix. Returning to the equation $R G(t,d) = G(T,D)$, we see that the $k+2$ columns of $G_k(t,d)$ are eigenvectors of $R$, any $k$ of which are linearly independent. If $\lambda_1, \dots, \lambda_{k+2}$ are the corresponding eigenvalues, then it follows that the multiset formed by any $k$ of these $\lambda_i$ equals the multiset of eigenvalues of $R$. In particular for each $i$ satisfying $1 \leq i \leq k$, we have:
  \[ \{\lambda_1, \dots,\lambda_k\} = \{\lambda_1, \dots, \lambda_{i-1}, \lambda_{k+1},\lambda_{i+1}, \dots, \lambda_{k}\} \quad \text{as multisets}.\]
  This is possible only if $\lambda_1 = \dots = \lambda_k$, whence $R$ is a scalar matrix.
  \end{proof}
We are finally ready to prove Theorem \ref{main_thm}.
\begin{proof} [Proof of Theorem \ref{main_thm}]
Let $\mG$ denote the inverse image of Aut$(\mC_k) = \rho'(PGL_2(\bK))$ under the canonical epimorphism $\pi:GL_k(\bK) \to PGL_k(\bK)$.
We note that the group $\bK^{\times} = \text{ker}(\pi)$ of scalar matrices is contained in the center of $\mG$. Therefore, $\mG$ is a certain central extension of $PGL_2(\bK) \simeq \rho'(PGL_2(\bK))$ by $\bK^{\times}$. We will prove that each $\sim$-equivalence class in $S_{k,n}$ is in bijective correspondence with $\mG$.
Suppose $R \in GL(k, \bF_q)$ satisfies $ R \, G_k(t,d) = G_k(\tau,\delta)$, for some $G(t,d)$ and $G(\tau,\delta)$. Since $n \geq k+2$, we can use Lemma \ref{k+2}  to conclude that  $[R] \in $Aut$(\mC_k)$, i.e. $R \in \mG$. \\

It now suffices to prove that $\mG$ acts freely on $S_{k,n}$. Suppose $R \, G(t,d) = G(t ,d)$ where $R \in \mG$. By Theorem \ref{aut_rnc}, there is a unique $g \in PGL_2(\bK)$ such that  $\pi(R) = \rho'(\pi(g))$. By Proposition \ref{equivariance}, the  Mobius transformation $g(t) =  \tfrac{\gamma + \delta t}{\alpha+\beta t}$ fixes $n \geq 3$ points $\{[1,t_i] \,:\, i = 1 \dots n\}$ of $\bP^1$. Therefore, $g$ (and hence $[R] = \rho'([g])$) is the identity transformation. It follows that $R$ is a scalar matrix. The equation  $R \, G(t,d) = G(t ,d)$ now implies $R$ is the identity matrix.
\end{proof}

A more general question is to compute the number $\gamma(k,n)$ of distinct MDS codes of length $n$ and dimension $k$ over $\mathbb{F}_q$. When $q$ is a prime, then $n \le q+1$ and if $n=q+1$, the codes are GRS \cite{Ball1,Ball2}. But when $q=9$, there is an example of a $10$-arc which is not an NRC \cite{Glynn1986}. Surprisingly, it is possible to describe the asymptotic behaviour of the number of distinct MDS codes. More precisely, in \cite{kaipa} some results were obtained recently about the asymptotic behaviour of $\gamma(k,n)$ where it was shown that:
\begin{equation}
\gamma(k,n)=q^\delta+(1-N)q^{\delta-1}+a_2q^{\delta-2}+\mathcal O\left(q^{\delta-3}\right),
\end{equation}
where $\delta=k(n-k)$, $N=\binom{n}{k}$ and
$$a_2=N k(n-k)\left(\dfrac{k^2-nk+n+3}{2(k+1)(n-k+1)}\right)+\dfrac{N^2}{2}-5\dfrac{N}{2}+2.$$
In these asymptotic formulae, it is assumed that $n$ and $k$ are fixed, while $q$ tends to infinity. We can use Theorem \ref{thm:countGRS} to obtain a similar asymptotic formula for the number of GRS codes:

\begin{proposition}\label{prop:countGRSasympt}
Let $\mathbb{F}_q$ be the finite field with $q$ elements and choose natural numbers $k,n$ such that $4 \le k+2 \le n \le q+1$. For fixed $k$ and $n$ and $q$ tending to infinity, we have:
\begin{multline} \gamma_{GRS}(k,n) = q^{2n-4}-\frac{(n-2)(n+1)}{2}q^{2n-5} \\ \\ +\frac{(n-2)(3n^3-4n^2+n-24)}{24}q^{2n-6}+\mathcal O\left(q^{2n-7}\right).
\end{multline}
\end{proposition}
\begin{proof}
On the one hand, it is easy to see that
\begin{equation}\label{eq:firstterm}
(q-1)^{n-2}=q^{n-2}-(n-2)q^{n-3}+\frac{(n-2)(n-3)}{2}q^{n-4}+\mathcal O\left(q^{n-5}\right),
\end{equation}
while on the other hand,
$$(q-1)\cdot (q-2) \cdots (q-n+2)=q^{n-2}-\sum_{i} i \ q^{n-3}+ \sum_{i<j}i\cdot j \ q^{n-4}+\mathcal O\left(q^{n-5}\right),$$ where in both summations the variables $i$ and $j$ vary between $1$ and $n-2$. The coefficient of $q^{n-3}$ is then easily determined, since $$\sum_{i=1}^{n-2} i =\dfrac{(n-1)(n-2)}{2}.$$ The determination of the coefficient of $q^{n-4}$ is somewhat more involved, but we have
$$\sum_{i<j}i\cdot j = \sum_{i=1}^{n-2} i \sum_{j=i+1}^{n-2} j = \sum_{i=1}^{n-2} i \left[ \dfrac{(n-1)(n-2)}{2}-\dfrac{(i+1)i}{2} \right].$$ If we apply the well-known summation formulae for squares and cubes to this expression and simplify, we see that
\begin{equation}\label{eq:secondterm}
\begin{split}
(q-1)\cdot (q-2) \cdots (q-n+2)=q^{n-2}-\frac{(n-1)(n-2)}{2}q^{n-3}+ \\ \\ \frac{(n-1)(n-2)(n-3)(3n-4)}{24}q^{n-4}+\mathcal O\left(q^{n-5}\right).
\end{split}
\end{equation}
Since by Theorem \ref{thm:countGRS} the number of distinct $[n,k]$ GRS codes over $\mathbb{F}_q$ is just the product of the left-hand sides in equations \eqref{eq:firstterm} and \eqref{eq:secondterm}, we readily obtain the desired asymptotic formula by multiplying the right-hand sides occurring in these equations.
\end{proof}

\begin{remark}
The first term $q^{2n-4}$ in the formula in Proposition \ref{prop:countGRSasympt} corresponds to the fact observed in \cite{dur} that the GRS codes form an algebraic subset of dimension $2n-4$ in the $k(n-k)$ dimensional affine space. 
\end{remark}

\begin{remark}
Exploiting the structure of the coefficient of the asymptotic expansion of $\gamma_{GRS}(n,k)$, it is easy to ascertain that the coefficient of $q^{2n-4-i}$ will be a polynomial in $n$ (with rational coefficients) of degree $2i$. Using this fact, one can easily determine further coefficients. For the interested reader we state the result one obtains when computing the coefficient $a_{2n-7}$ of $q^{2n-7}$ in the asymptotic expansion: $$a_{2n-7}=-\dfrac{(n-2)(n-3)(n^4-2n^3+7n^2-14n-16)}{48}.$$
\end{remark}

\begin{remark}
We note that $\gamma_{GRS}(n,k) = \gamma(n,2)$ for $4 \leq n \leq q+1$.
To see this, we observe that $\gamma_{GRS}(n,k)$ is independent of $k$ for $4 \leq k+2 \leq n \leq q+1$ by Theorem \ref{thm:countGRS}. Thus $\gamma_{GRS}(n,k) = \gamma_{GRS}(n,2)$ for $4 \leq n \leq q+1$.
Any generator matrix for a $2$-dimensional MDS code is clearly of the form $G_2(t,d)$ (see \eqref{eq:gen_mat}), and hence every $2$-dimensional MDS code is GRS. Thus $\gamma_{GRS}(n,2) = \gamma(n,2)$.
\end{remark}

\section{Counting GRS codes of length $q+2$.}

It is well known that for $q$ even, there exist MDS codes of length $q+2$ and dimension $3$. Such codes arise from combinatorial structures known as hyperovals, which we will define in a moment. First of all for a natural number $n$, one defines a (planar) $n$-arc to be a set of $n$ mutually distinct points in the projective plane $\mathbb{P}^2(\mathbb{F}_q)$ with the property that no three points in the $n$-arc lie on a line. In case $n=q+1$, such an $n$-arc is called an oval, while if $n=q+2$ it is called a hyperoval. It is well known that $n$-arcs give rise to $[n,3]$ MDS codes. A generator matrix of a code can be obtained from an $n$-arc the following way: first one orders the projective points in the $n$-arc, then one chooses representatives in $\mathbb{F}_q^3$ of the projective points and finally one uses these $n$ ordered elements from $\mathbb{F}_q^3$ as the columns of a generator matrix of a code. The resulting codes are easily seen to be $[n,3]$ MDS codes by the defining properties of an $n$-arc. Note that there exists a more general definition of arcs in $\mathbb{P}^m(\mathbb{F}_q)$, but in this article we will always assume that $m=2$ and assume this implicitly when speaking about arcs.

In 1947 R.C.~Bose proved that the maximal size of any arc in a projective plane of odd order $q$ is $q+1$, and if $q$ is even the maximum may be $q+2$, see pp.~149 in \cite{Dembowski}. In odd characteristic, Segre (see Sections 173 and 174 in \cite{segre}) proved that any $(q+1)$-arc, that is to say any oval, in fact consists of the $q+1$ rational points lying on a non-degenerate conic. By slight abuse of terminology such ovals are simply called conics as well. In even characteristic the situation is different: not all ovals can be described as a conic, but other ovals exist. It is still true though that the $q+1$ rational points on a non-degenerate conic form an oval. In even characteristic one type of hyperoval can be obtained by taking the $q+1$ rational points on a conic and then adding the nucleus which is the common intersection point of all tangent lines of this conic \cite{Glynn2, HKT}. We will use the word hyperconic for such a hyperoval. In general not all hyperovals can be described as a hyperconic, but for $q=2$, $q=4$ and $q=8$, all hyperovals are known to be hyperconics (page 290 \cite{segre}). Note that if an $n$-arc is contained in a conic, the MDS code corresponding to this $n$-arc actually is a GRS code. Following this relation, we define by analogy the $[q+2,3]$ MDS codes obtained from hyperconics to be GRS codes as well. From now on in this section, we will assume that $q$ is even. Our goal is to count the number of distinct $[q+2,3]$ GRS codes. Since any GRS code with parameters $[q+2,3]$ corresponds to a hyperconic, we can obtain all of these codes by extending $[q+1,3]$ GRS codes. Geometrically, this means that we simply are adding the nucleus to the conic. Conics and hyperconics play an important role in the theory of GRS codes. A generator matrix for a GRS code of length $q+1$ can be obtained by putting the $q+1$ points of a conic as columns of a matrix. Adding an additional ``nucleus" column to, yields a generator matrix of a length $q+2$ GRS code.

We are now ready to count the number of distinct $[q+2,3]$ GRS codes:
\begin{theorem}\label{thm:countGRS2}
Let $q=2^e$ for an integer $e \ge 3$. The number of distinct GRS codes of length $q+2$ and dimension $3$ equals $$\gamma_{GRS}(3,q+2)=(q+2)(q-1)^{q+1}(q-2)!.$$
\end{theorem}
\begin{proof}
Since any GRS code with parameters $[q+2,3]$ corresponds to a hyperconic, we can obtain these codes by extending $[q+1,3]$ GRS codes. We already know from Theorem \ref{thm:countGRS} that the number of distinct GRS codes of length $q+1$ equals $(q-1)^q(q-2)!$. Adding a``nucleus" column can be done in $q+2$ distinct positions and also gives rise to an additional column multiplier, giving in total $(q+2)(q-1)^{q+1}(q-2)!$ GRS codes. By definition, all GRS codes of length $q+2$ are obtained in this way. Some of these codes could in principle be the same, but if $q>4$, the group acting on the conic plus nucleus hyperoval (a nonic in \cite{Glynn3}), fixes the nucleus and is sharply $3$-transitive on the remaining points. This means that all $(q+2)(q-1)^{q+1}(q-2)!$ GRS codes obtained above are distinct.
\end{proof}

\begin{remark}
Theorem \ref{thm:countGRS2} does not cover the case that $q=4$, as it was needed in the proof that $q>4$. Indeed in case $q=4$, the (linear) group of the conic plus nucleus is sharply $4$-transitive on all the $q+2=6$ points. The right formula for $q=4$ is therefore given by $$\gamma_{GRS}(3,6)=(q-1)^{q+1}(q-2)!=486.$$

Alternatively, in the next section we shall among others quote a known expression for the number of distinct $[6,3]$ MDS codes. Combining this with Segre's observation that for $q=4$ any hyperoval is a hyperconic, this expression will give the desired number $486$ directly.
\end{remark}

\section{A further investigation of GRS and MDS codes of dimension three}
The formulae given in Theorems \ref{thm:countGRS} and \ref{thm:countGRS2} can be tested by comparing them to other results regarding the number of arcs or equivalently MDS codes as given in \cite{glynn,iam-skor-sor}. There it is stated that $\gamma(3,6)$, the number of distinct $[6,3]$ MDS codes over $\mathbb{F}_q$, equals
$$(q-1)^5( q - 2)(q - 3)(q^2 - 9q + 21).$$
For $q=5$, this formula combined with Theorem \ref{thm:countGRS} implies that any $[6,3]$ MDS code over $\mathbb{F}_5$ is a GRS code, since there exist exactly $6144$ distinct ones of either type. This is in accordance with the aforementioned result by Segre, that in odd characteristic any oval is a conic. For $n=7,8,9$ there exist closed formulae for $\gamma(3,n)$, the number of distinct $[n,3]$ MDS codes, as well \cite{glynn,iam-skor-sor,rolland}. To state these formulae, it is convenient to define the following functions:
$$
\begin{array}{l}
a(q)=1 \ \makebox{if $q$ is a power of $2$, and $a(q)=0$ otherwise,}\\
\\
b(q)=|\{x \in \mathbb{F}_q \, | \, x^2+x+1=0\}|\\
\\
c(q)=1 \ \makebox{if $q$ is a power of $3$, and $c(q)=0$ otherwise,}\\
\\
d(q)=|\{x \in \mathbb{F}_q \, | \, x^2+x-1=0\}|\\
\\
e(q)=|\{x \in \mathbb{F}_q \, | \, x^2+1=0\}|\\
\end{array}
$$
The number of distinct $[7,3]$ MDS codes over $\mathbb{F}_q$ is then given by \cite{rolland}:
$$(q- 1)^6 (( q-3) (q- 5) ( q^4 - 20q^3+ 148q^2 -468q+498) -30a(q)).$$
The number of distinct $[8,3]$ MDS codes over $\mathbb{F}_q$ is given by \cite{glynn}:
$$(q - 1)^7((q - 5)(q^7 - 43q^6 + 788q^5-7937q^4 + 47097q^3- 162834q^2
+299280q - 222960)$$ $$ - 240(q^2 - 20q + 78)a(q) +840b(q)).$$
Finally from \cite{iam-skor-sor}, we quote the result that the number of distinct $[9,3]$ MDS codes is given by
$$(q - 1)^8(q^{10} - 75q^9+2530q^8 - 50466q^7 + 657739q^6 - 5835825q^5
+35563770q^4$$ $$- 146288034q^3 + 386490120q^2
-588513120q + 389442480
-1080(q^4 - 47q^3 + 807q^2$$ $$ - 5921q +15134)a(q)
+840(9q^2 - 243q + 1684)b(q)
+30240(-9c(q) + 9d(q) + 2e(q))).$$

These formulae and Theorems \ref{thm:countGRS} and \ref{thm:countGRS2} give rise to the following table in which for given $q$ and $n$ we count the number of $[n,3]$ GRS and MDS codes over $\mathbb{F}_q$. The two mentioned cases of the number of distinct $[10,3]$ MDS codes are not covered by the above formulae. However, since Segre showed that for $q=4$ and $q=8$ any hyperoval is a hyperconic, these numbers can be computed using Theorems \ref{thm:countGRS} and \ref{thm:countGRS2}.

\begin{center}
\begin{figure}\notag
\begin{tabular}{|c|c|c|c|}
\hline
$q$ & $n$ & $\#$ GRS & $\#$ MDS \\
\hline
$4$ & $6$ & $486$ & $486$ \\
\hline
$5$ & $6$ & $6144$ & $6144$ \\
\hline
$7$ & $6$ & $466560$ & $1088640$ \\
\hline
$7$ & $7$ & $5598720$ & $5598720$ \\
\hline
$7$ & $8$ & $33592320$ & $33592320$ \\
\hline
$8$ & $6$ & $2016840$ & $6554730$ \\
\hline
$8$ & $7$ & $42353640$ & $141178800$ \\ 
\hline
$8$ & $8$ & $592950960$ & $2964754800$ \\ 
\hline
$8$ & $9$ & $4150656720$ & $41506567200$ \\ 
\hline
$8$ & $10$ & $290545970400$ & $290545970400$ \\
\hline
$9$ & $6$ & $6881280$ & $28901376$ \\
\hline
$9$ & $7$ & $220200960$ & $1604321280$ \\
\hline
$9$ & $8$ & $5284823040$ & $15854469120$ \\ 
\hline
$9$ & $9$ & $84557168640$ & $84557168640$ \\
\hline
$9$ & $10$ & $676457349120$ & $676457349120$ \\
\hline
\end{tabular}
\vspace{.9ex}

\text{Table 1: Number of GRS and MDS codes of small length and dimension $3$.}
\end{figure}
\end{center}

There are some observations we would like to make: in the first place, Segre's observation that in odd characteristic any oval is a conic, is confirmed by the above table. Of course also his result that for $q=4$ and $q=8$ any hyperoval is a hyperconic is reflected in the table, but we actually used this result to calculate the number of MDS codes in these situations. Another observation is that for odd $q$ the number of $[q,3]$ GRS and MDS codes are the same. This is a reflection of the fact that Segre showed that in odd characteristic any $q$-arc can be extended (uniquely if $q>3$) to an oval. The following observation is less trivial and in the remaining part of the article we will explain it.
\begin{observation}\label{obs:int}
Let $q=8$. Then the number of $[9,3]$ (resp. $[8,3]$) MDS codes is $10$ (resp. $5$) times the number of GRS codes of this length and dimension.
\end{observation}
The point of the observation is that it is somewhat unexpected that the ratio of the number of MDS and GRS codes is an integer in some cases. We will give an explanation of this in the following. For $3 \le n \le q+2$ let us denote by $\mathcal{MDS}(n)$ (resp. $\mathcal{GRS}(n)$) the set of $[n,3]$ MDS codes (resp. GRS codes) over $\mathbb{F}_q$. We now define the following map:
\begin{definition}
Let $q=2^e$ with $e \ge 3$ and choose $1 \le r \le q+1$. Then we define the map
$$
P_{r}: \mathcal{GRS}(q+2) \rightarrow \mathcal{MDS}(q+2-r)$$ by putting for $C \in \mathcal{GRS}(q+2)$ $$P_r(C):=\{(c_1,\dots,c_{q+2-r}) \, | \, \exists \, c_{q+2-r+1},\dots,c_{q+2} \ \makebox{s.t.} \ (c_1,\dots,c_{q+2}) \in C\}.
$$
\end{definition}

In words: $P_r$ punctures a code in the last $r$ coordinates. The reason that we only consider this map for even $q$ and for length $q+2$ GRS codes, is that in all other cases, a punctured GRS code is a GRS code again. We will later see that this is not true in the case we are considering. First we collect some properties of the map $P_r$ in the following:

\begin{lemma}\label{lem:regularity}
Suppose that $q+2-r \ge 7$, then the map $P_r$ is a $r!(q-1)^r$ to one map.
\end{lemma}
\begin{proof}
Let $C \in \mathcal{GRS}(q+2)$ and denote by $G_r$ denote a generator matrix of the code $C_r:=P_r(C)$. We wish to determine the number of possibilities for $D \in \mathcal{GRS}(q+2)$ such that $P_r(D)=C_r$. Equivalently we wish to describe the codes $D\in \mathcal{GRS}(q+2)$ that have a generator matrix of the form $(G_r \, B)$, with $B$ a $3 \times r$ matrix.

The $q+2-r$ columns of $G_r$ give rise to $q+2-r$ projective points in $\mathbb{P}^2(\mathbb{F}_q)$ lying on a hyperconic. We claim that this hyperconic is unique: indeed if two hyperconics have at least $q+2-r$ points in common, then discarding the nuclei of these hyperconics, one obtains two conics having at least $q-r$ points in common. Since by assumption $q-r \ge 5$ these conics are identical by Theorem \ref{Castelnuovo} applied to the case $k=3$. This implies that the original hyperconics were identical as well.

This shows that if $P_r(D)=C_r$ and $D$ has generator matrix $(G_r \, B)$, then the set of $r$ projective points that the columns of $B$ give rise to, is uniquely determined. This leaves $r! (q-1)^r$ possibilities for $B$, since each projective point has $q-1$ representatives and the ordering of the $r$ representatives as columns of $B$ can be chosen freely.
\end{proof}

An $n$-arc is called complete, if it cannot be extended to a $(n+1)$-arc. If $q$ is odd, it is for example known that no complete $q$-arcs exist \cite{Hirschfeld1998}. Similarly if $q$ is even, no complete $(q+1)$-arcs exist \cite[Ch.8]{Hirschfeld1998}. The concept of completeness is instrumental in explaining Observation \ref{obs:int} and plays a role in the following:

\begin{lemma}\label{lem:surj}
Suppose that any hyperoval is a hyperconic and moreover that for $n = q+2-r,\dots,q+1$ any $n$-arc can be extended to an $(n+1)$-arc. Then the map $P_r$ is surjective.
\end{lemma}
\begin{proof}
From the assumptions, we see that any $n$-arc with $q+2-r \le n \le q+1$ can be extended. Inductively, this means that any $(q+2-r)$-arc can be extended to a hyperoval, which in turn was assumed to be a hyperconic. This implies that any $[q+2-r,3]$ MDS code can be extended to a $[q+2,3]$ GRS code. A direct consequence is that the map $P_r$ is surjective.
\end{proof}

We are now ready to explain Observation \ref{obs:int}.

\begin{proposition}
Let $q=8$ and $r=1,2,3$. Then we have
$$\dfrac{|\mathcal{MDS}(q+2-r)|}{|\mathcal{GRS}(q+2-r)|}=\dfrac{q+2}{r}.$$
\end{proposition}
\begin{proof}
First of all note that by Lemma \ref{lem:regularity} (which we may apply, since $q=8$ and $r\le 3$) we have that
\begin{equation}\label{eq:cardim}
|P_r(\mathcal{GRS}(q+2))|=\dfrac{|\mathcal{GRS}(q+2)|}{r!(q-1)^r}
\end{equation}
From Lemma \ref{lem:surj} and the fact that there do not exist complete $n$-arcs over $\mathbb{F}_8$ with $n=7,8$ or $9$ (pages 285 and 290 in \cite{segre}), we may conclude that $P_r(\mathcal{GRS}(q+2))=\mathcal{MDS}(q+2-r).$ Combining this with \eqref{eq:cardim} and Theorems \ref{thm:countGRS} and \ref{thm:countGRS2}, we see that
$$\dfrac{|\mathcal{MDS}(q+2-r)|}{|\mathcal{GRS}(q+2-r)|}=\dfrac{|\mathcal{GRS}(q+2)|}{r!(q-1)^r\cdot |\mathcal{GRS}(q+2-r)|}=\dfrac{q+2}{r}.$$
\end{proof}



\end{document}